\documentclass[11pt]{article}
\usepackage[margin=1in]{geometry}
\usepackage{times}
\usepackage{amsfonts}
\usepackage{amsmath,amsthm,amssymb}
\usepackage{latexsym}
\usepackage{color}
\usepackage{epic}
\usepackage{epsfig}
\usepackage[hang,raggedright]{subfigure}
\usepackage{wrapfig}
\usepackage[utf8]{inputenc} 
\usepackage[T1]{fontenc}    
\usepackage{url}            
\usepackage{booktabs}       
\usepackage{amsfonts}       
\usepackage{nicefrac}       
\usepackage{microtype}      

\usepackage[justification=centering]{caption}
\allowdisplaybreaks[1]
\usepackage{algorithm}
\usepackage{algorithmic}
\usepackage{setspace}
\usepackage{float}

\def\ddp{\qopname\relax o{DP}}
\makeatletter
\newcommand*{\rom}[1]{\expandafter\@slowromancap\romannumeral #1@}
\makeatother

\usepackage{array}

\newcolumntype{L}[1]{>{\raggedright\let\newline\\\arraybackslash\hspace{0pt}}m{#1}}
\newcolumntype{C}[1]{>{\centering\let\newline\\\arraybackslash\hspace{0pt}}m{#1}}
\newcolumntype{R}[1]{>{\raggedleft\let\newline\\\arraybackslash\hspace{0pt}}m{#1}}

\def\en{\texttt{MaxEn}}
\def\ca{\texttt{CARD}}
\def\ge{\texttt{ColG}}
\def\ri{\texttt{RigorOPT}}

\newtheorem{theorem}{Theorem}[section]

\newtheorem{lemma}[theorem]{Lemma}
\newtheorem{proposition}[theorem]{Proposition}

\newtheorem{definition}[theorem]{Definition}

\newtheorem*{conjecture*}{Conjecture}
\newtheoremstyle{nonindented}{1ex}{1ex}{}{}{\bfseries}{.}{.5em}{}
\newtheoremstyle{indented}{1ex}{1ex}{\itshape\addtolength{\leftskip}{0.6cm}\addtolength{\rightskip}{0.6cm}}{}{\bfseries}{.}{.5em}{}
\theoremstyle{nonindented}
\theoremstyle{indented}
\theoremstyle{plain}

\renewcommand{\tilde}{\widetilde}

\DeclareMathOperator{\poly}{poly}

\def\min{\qopname\relax n{min}}
\def\max{\qopname\relax n{max}}

\def\Pr{\qopname\relax n{\mathbf{Pr}}}
\def\Ex{\qopname\relax n{\mathbf{E}}}

\newcommand{\RR}{\mathbb{R}}

\def\M{\mathcal{M}}
\def\P{\mathcal{P}}

\def\S{\mathcal{S}}

\def\X{\mathcal{X}}

\newcommand{\maxi}[1]{\mbox{maximize} & {#1 } & \\}

\newcommand{\st}{\mbox{subject to} }

\newcommand{\con}[1]{&#1 & \\}
\newcommand{\qcon}[2]{&#1, & \mbox{for } #2.  \\}
\newenvironment{lp}{\begin{equation}  \begin{array}{lll}}{\end{array}\end{equation}}
\newenvironment{lp*}{\begin{equation*}  \begin{array}{lll}}{\end{array}\end{equation*}}

\begin{document}
	
	\title{Mitigating the Curse of Correlation in Security Games by  Entropy Maximization}
	\author{
		\begin{tabular}{ccccccc}
			\centering
			Haifeng Xu &\quad &  Shaddin Dughmi  &\quad&  Milind Tambe &\qquad &   Venil Loyd Noronha \tabularnewline
		\end{tabular} \vspace{2mm} 	\\
		 University of Southern California, Los Angeles, USA \vspace{1mm} \\
		\{haifengx,shaddin,tambe,vnoronha\}@usc.edu
	}
	

	\maketitle


\begin{abstract}
In Stackelberg security games, a defender seeks to randomly allocate  limited security resources to protect critical targets from an attack.  In this paper, we study a fundamental, yet underexplored, phenomenon in security games, which we term the \emph{Curse of Correlation} (CoC). Specifically, we observe that there are inevitable correlations among the protection status of different targets.  Such correlation is a crucial concern, especially in \emph{spatio-temporal} domains like conservation area patrolling, where attackers can surveil patrollers at certain areas and then infer their patrolling routes using such correlations. To mitigate this issue, we propose to design entropy-maximizing defending strategies for spatio-temporal security games, which frequently suffer from CoC. 
We prove that the  problem is \#P-hard in general.  However, it admits efficient algorithms in well-motivated special settings. Our experiments show significant advantages of max-entropy algorithms over previous algorithms. 
A scalable implementation of our algorithm is currently under pre-deployment testing for integration into FAMS software to improve the scheduling of US federal air marshals.

\end{abstract}

\section{Introduction}  
The security game is played between a defender and an attacker where the defender’s goal is to \emph{randomly} allocate a limited number of security resources to protect targets from attack
\cite{Tambe2011,Blum2014,Yin2015}. Standard models assume that the attacker only knows the defender's randomized strategy, but is oblivious to the sampled real-time allocation. However, this assumption may fail since in many situations the attacker can  \emph{partially} observe the defender's real-time allocation. Such partial observation can be utilized to infer extra information about the overall strategy. This is particularly the case in games played out in space and time, a.k.a. \emph{spatio-temporal} security games, which are also the primary focus of this work. For example, it has been reported that in wildlife protection domains, some poachers partially monitor rangers' patrolling activities  and then make their poaching plans based on their observations \cite{Nyirenda2012,Moreto2013}. Similar issues could happen when optimizing the patrolling on graphs \cite{noa08a,Gatti2} and randomized scheduling of air marshals \cite{News1}.

We observe that the randomized allocation of limited security resources creates \emph{inherent} correlations within protection statuses of different targets -- the coverage of some targets implies that other targets are not protected. Such correlation allows the attacker to utilize his partial observation at some targets to infer information about other targets' protection status, a phenomenon which we term the \emph{Curse of Correlation} (CoC).  We prove that such correlation is inevitable, and may cause significant loss if not addressed properly. The ideal approach for mitigating CoC is to come up with an accurate model to capture the attacker's partial observations -- which is also called an  \emph{information leakage} model  \cite{Xu15} --  and then compute the optimal defender strategy tailored specifically to this leakage model. For example, one possibility is to model the interaction as an extensive-form game and then solve the game. However, we argue that this approach suffers from several critical drawbacks: 1. lack of  accurate leakage models in practice; 2. limited scalability;  3. vulnerability to the attacker's strategic manipulations. These barriers motivate our adoption of a more robust (though inevitably more conservative) approach.  Particularly, we propose to adopt the ``most random'' defending strategy, or more formally, the strategy that, subject to optimizing the usual objective under required constraints, maximizes (Shannon) \emph{entropy}. Intuitively, such a strategy could be resistant to partial leakage due to its extreme randomness/unpredictability. 

To that end, this paper offers the following contributions. First, we formally study the Curse of Correlation (CoC) phenomenon in security games and illustrate the importance of handling CoC, particularly in spatio-temporal domains. Second, we propose to adopt the defending strategy with maximum entropy and illustrates its advantages when compared to the idealized optimal solution tailored to a specific leakage model. Experiments show that our approach is close to the idealized optimality but are more robust to leakage model misspecification and attacker's strategic manipulations.  Third, we design novel algorithms to sample defending strategies with maximum entropy in spatio-temporal security games. We prove that the exact max-entropy defending strategy is  \#P-hard to compute in general, but admits polynomial-time algorithms in well-motivated special settings as well as efficient heuristic algorithms. Experiments show  that our algorithms are far more scalable than the idealized optimal approach. \emph{Notably, a scalable implementation of our algorithm is currently under pre-deployment testing for integration into the Federal Air Marshal Service (FAMS) software to improve the scheduling of US federal air marshals.}

\vspace{2mm}
\noindent {\bf Related Work.} Most relevant to ours is the recent work of \cite{Xu15}, which focuses on a  basic security setting with \emph{no} spatio-temporal structures, \emph{no} scheduling constraints and a \emph{specific}  information leakage model, i.e., only one target leaks information. Despite such simplifications, they prove that computing the optimal defender strategy is NP-hard, and proposes an \emph{exponential-time} algorithm. Our work differs from \cite{Xu15} in two key aspects. First, \cite{Xu15} raises the issue of information leakage in security games whereas we argue that the reason that  information leakage hurts the defender is the underlying correlation among targets and, more importantly, formally prove that such correlation is \emph{inevitable}. This is also why we term the phenomenon the curse of correlation.  Second, the security settings we consider in this paper are more complicated with spatio-temporal structures, scheduling constraints\footnote{Such scheduling constraints are crucial for important applications in many domains such as FAMS, wildlife conservation, coast guard patrolling, etc.} and more leaking targets, thus is even more intractable. We are not aware of any efficient way to extend and adapt the (already exponential-time) algorithm in \cite{Xu15} to our setting.

Alon et al. \cite{Alon2013} study information leakage in \emph{normal-form} zero-sum games and illustrate the difficulty of handling leakage by exhibiting NP-hardness results in several model variants. However, the specific leakage models in \cite{Alon2013}  do not directly fit into the security game applications we consider. In principle, the leakage problem can also be represented as an extensive-form game (EFG) and there has been significant progress in solving general-purpose large EFGs recently \cite{EFG0,EFG1,EFG2,EFG3}. However, we did not take this approach because the size of information sets in our game increases exponentially in the number of security resources, time steps and possibly leaking targets. This very quickly makes our problem intractable (see more discussions in later sections).   
This work also relates to the line of research on adversarial patrolling games (APGs) \cite{noa08a,Gatti2,Bosansky,Alpern2011,Vorobeychik2012}. Our discussion here cannot cover the rich body of work on APGs; We refer curious readers to \cite{Gatti2} for a good overview. APGs also consider the attacker's real-time observations, however our settings differ from APGs in several aspects: 1. the defender in APGs \emph{typically} has only one patroller and the attacker has \emph{full knowledge} of the defender's movements, while in our setting the defender has many security resources and the attacker can only observe  a small subset of targets; 2. APGs assume that the attacker takes time to complete an attack, while attacks in our setting are instantaneous. These important differences make the algorithms for APGs inapplicable to our settings.

\section{Preliminaries}\label{sec:prelim}
The security game is  played between a defender and an attacker. The defender aims to use limited security resources to protect $n$ targets, denoted by set $[n]$, from the attacker's attack. Any \emph{defender pure strategy} $S\subseteq [n]$ is a subset of targets that can be feasibly covered by the security resources. A more convenient representation of $S$ is a binary vector $\mathbf{s} = (s_1,...,s_n)^T \in \{0,1 \}^n$, indicating whether each target is covered or not. We will use $S$ or $\mathbf{s}$ interchangeably throughout, and the meaning should be clear from the context.  Let $\S$ denote the set of all defender pure strategies, whose size is typically \emph{exponential} in the game representation \cite{Tambe2011}.  For example, $\S$ could consist of all subsets of $[n]$ of a fixed size, which corresponds to the setting of no resource allocation constraints. However, in reality, $\S$ usually has more structures due to resource allocation constraints.
A  defender mixed strategy is a distribution $\mathbf{p}$ over $\S$.  The attacker pure strategy is a target $i \in [n]$ which he chooses to attack. The concrete payoff structure of the game is not particularly important  for the purpose of this paper. We only mention that, given that target $i$ is attacked, the defender and attacker both get a utility depending on its importance  and marginal protection probability, described below.      



\vspace{1mm}
\noindent {\bf Marginal Probabilities and Implementability}. For any defender mixed strategy $\mathbf{p}$, $\sum_{\mathbf{s} \in \S} p_{\mathbf{s}} \cdot s_i =x_i  \in [0,1]$ is the marginal (coverage) probability of target $i$. Let  $\mathbf{x}=(x_1,$ $...x_n)^T = \sum_{\mathbf{s} \in \S} p_{\mathbf{s}} \cdot  \mathbf{s}  \in [0,1]^n$ be the \emph{marginal vector} of all targets. In this case, we say that mixed strategy $\mathbf{p}$ \emph{implements} marginal vector $\mathbf{x}$, and $\mathbf{x}$ is \emph{implementable} (by $\mathbf{p}$). Not every $\mathbf{x} \in [0,1]^n$ is implementable -- all the implementable $\mathbf{x}$'s form a polytope $\mathbf{conv}(\S)$, i.e., the convex hull of $\S$. 

There are generally many mixed strategies that implement the same marginal vector.  Due to the concern of computational efficiency, most algorithms for security games are designed to compute a mixed strategy of small support, which is evidently a poor choice in the presence of attacker surveillance.  A primary goal of this paper is to design mixed strategies that are robust to attacker partial observations. Next, we describe a particular  implementation for a marginal vector. 



\vspace{1mm}
\noindent {\bf The Max-Entropy Implementation}.  
Given any implementable $\mathbf{x}$, the \emph{max-entropy implementation} of $\mathbf{x}$ is the mixed strategy $\mathbf{p}$ of the maximum (Shannon) entropy among all mixed strategies that implement $\mathbf{x}$. More precisely, the max-entropy implementation of  $\mathbf{x}$ is the optimal solution to the following convex program (CP) with variables $p_{S}$.  

\begin{lp} \label{lp:maxentropy}
	\maxi{ -\sum_{S \in \S }p_{S} \log( p_{S} ) } 
	\st
	\qcon{ \sum_{S \in \S: i \in S} p_{S} = x_i}{i \in [n]}
	\con{ \sum_{S \in \S} p_{S} =1 }   
	\qcon{p_{S} \geq 0}{ S \in \S}   
\end{lp} 
An obvious challenge of solving CP \eqref{lp:maxentropy}  is that the optimal $\mathbf{p}^*$ typically has exponentially large support, thus cannot even be  written down explicitly in polynomial time.  It turns out that this can be overcome  via algorithms that efficiently samples  a set $S$ ``on the fly''. Therefore, we say an algorithm solves CP \eqref{lp:maxentropy} if it takes $\mathbf{x}$ as input and randomly samples  $S \in \S$  with probability $p_{S }^*$ where  $\mathbf{p}^*$ is the  optimal solution to CP  \eqref{lp:maxentropy}. It turns out that we can sample the max-entropy implementation in polynomial time as long as we can solve the  following \emph{generalized counting} problem over $\S$. 

\begin{definition}[Generalized Counting]
	Given any $\mathbf{\alpha} \in R_+^n$, compute $C(\alpha) = \sum_{S \in \S} \alpha_S$, where $\alpha_S = \prod_{i\in S} \alpha_i$. 
	\end{definition}

Note that $C(\mathbf{1})$ equals precisely the cardinality of $\S$; More generally, $C(\mathbf{\alpha})$ is a weighted count of the elements in $\S$ with weights $\alpha_S = \prod_{i\in S} \alpha_i$. 

\begin{lemma}\label{lem:EntroSample}(Adapted from \cite{singh14})
	The max-entropy implementation of any $\mathbf{x} \in \mathbf{conv}(\S)$ can be sampled in polynomial time
	if and only if the  generalized counting problem over $\S$ can be solved in polynomial time.
\end{lemma}
\vspace{-1mm}
As a result of Lemma \ref{lem:EntroSample}, to sample the max-entropy implementation of any marginal vector $\mathbf{x}$,  we only need to design a generalized counting algorithm for $\S$.  However, we do remark that  this does \emph{not} magically simplify the problem  since the generalized counting problem is computationally intractable for most combinatorial problems. In fact, \emph{very few} efficient algorithms are known, and the problem is \#P-hard even for counting simple combinatorial structures like perfect bipartite matchings \cite{Valiant1979}.

\section{The Curse of Correlation and A Remedy by Entropy Maximization}
\label{sec:curse}
In this section we illustrate how the attacker can first observe the protection status of a small set of targets and then utilize correlations to infer  the protection of other targets. To do so, the attacker does \emph{not} even need full knowledge of the defender's mixed strategy. Instead he only needs to know the correlations among the observed targets and other ones. 



\subsection{Hacking Security Games with One Bit of Information.}
\label{subsec:curse}
Our example is described  in the setting of conservation area protection, though we emphasize that the phenomenon it illustrates could occur in any security setting. The problem concerns designing rangers' patrolling schedules within a fixed time period, say, a day. This is usually modeled by discretizing the area into cells as well as discretizing the time. At the top of Figure \ref{fig:example}, we depict a concrete example with  $4$ cells to be protected at 3 time layers: morning, noon and afternoon. The numbers around each cell is the desired marginal coverage probabilities for each cell at each time (color thickness depicts the probability density).  
The defender has $2$ rangers, and seeks to randomize their patrolling  to achieve the required marginal probabilities.    

\begin{figure}[H]
	\centering

	\includegraphics[bb=40bp 100bp 740bp 490bp,clip,scale=0.31]{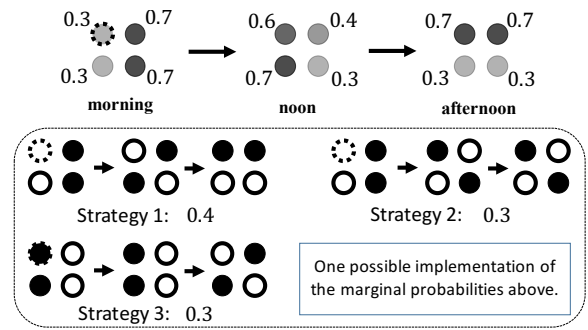}%
	\caption{\label{fig:example} \small Marginal Probabilities and Their Implementation.}

\end{figure}

To deploy a mixed strategies of small support, as done by traditional algorithms,  one can implement the marginal vector by mixing the three  pure strategies listed    at the bottom of  Figure \ref{fig:example}  (filled dots are
fully covered). Unfortunately, it turns out
that such an implementation is extremely vulnerable to the attacker's partial surveillance. For example, if the attacker can surveil the status of the top-left cell in the morning (i.e., the one with dashed boundary)  and prepare an  attack in the afternoon, he can always find a completely uncovered cell to attack. 
Specifically, if  the dashed cell is covered, this means \texttt{Strategy} $3$ is deployed and two particular cells are uncovered in the afternoon; Otherwise, either \texttt{Strategy} $1$ or $2$ is deployed, thus the bottom left cell is uncovered  in the afternoon for sure. So the attacker successfully identifies uncovered cells in the afternoon  by monitoring only \emph{one} cell  in the morning.  

The issue above is due to the inherent  correlation among the protection status of different targets when allocating a limited number of resources. Particularly, the coverage of  some targets must imply that some other targets are unprotected. The example illustrates how the attacker can take advantage of such correlation and infer a  significant amount of  information about the protection of other targets by monitoring even a single target. This is what we term the \emph{Curse of Correlation} (CoC) in security games. The following Proposition \ref{thm:EntroLoss} illustrates this phenomenon in a more formal sense.

We begin with a few notations. Recall that any mixed strategy $\mathbf{p}$ is a distribution over $\S$. Equivalently, we can view a mixed strategy as a random binary vector $\mathbf{X} =(X_1,...,X_n) \in \{0,1\}^n$ satisfying  $\Pr(X = \mathbf{s}) = p_{\mathbf{s}}$. Here, $X_i \in \{0,1\}$ denotes the random protection status of target $i$, and $\Pr(X_i = 1) = x_i$ is the marginal coverage probability.  For any $X_i$, let $H(X_i) = x_i \log x_i + (1-x_i) \log(1-x_i)$ denote its Shannon entropy.  Note that, $X_i$'s are correlated.  Let $X_i|X_k$ denote $X_i$ conditioned on $X_k$ and $H(X_i|X_k)$ denote its conditional entropy.  We say target $k$ is \emph{trivial} if $ \Pr(X_k \mbox{\,=\,}1) = 0$ or $1$; Otherwise, $k$ is \emph{non-trivial}. Obviously, the attacker infers no information about other targets by monitoring a trivial target.  The following proposition shows that  if any non-trivial target $k$ is monitored,  the attacker can always infer information about the protection of other targets.\footnote{Here we assume that all resources are fully used, and do not consider the  (unreasonable) situations in which certain security resources are sometimes underused or idle.}   

\begin{proposition}\label{thm:EntroLoss}
	For any non-trivial target $k$, we have $$
	\mathbf{E}_{X_k} \big[ \sum_{i  \not = k}  H(X_i|X_k) \big] < \sum_{i \not = k} H(X_i)$$.
\end{proposition}
\begin{proof}
	Let $\Pr(X_i = 1|X_k = 1) = x_i^1$ and $\Pr(X_i = 1|X_k = 0) = x_i^0$. Since $\Ex_{X_k} [\Pr(X_i = 1|X_k)] = \Pr(X_i = 1) = x_i$, we have $x_i =  x_i^1 \cdot \Pr(X_k = 1) + x_i^0 \cdot \Pr(X_k = 0) $. Since $H(X_i)$ is strictly concave w.r.t. $x_i$, we have 
	\begin{equation}\label{eqn:posteriorEqu}
	\Ex_{X_k} H(X_i|X_k) \geq H(X_i).
	\end{equation}
	Summing over all $i \not = k$, we get $ \sum_{i\in [n]: i \not = k} \Ex_{X_k} H(X_i|X_k) \geq \sum_{i\in [n]: i \not = k}H(X_i)$. We now argue that the ``=" \emph{cannot} hold. Prove by contradiction. Note that target $i$ is non-trivial, therefore $\Pr(X_k) \not = 0,1$.  Since $H(X_i)$'s are strictly concave, if the ``=" holds, then we must  have $\Pr(X_i = 1|X_k = 1) = \Pr(X_i = 1|X_k = 0) = x_i$ for any $i \not = k$. However, this implies that target $i$ is trivial, i.e., $i$ is either fully protected or unprotected. Otherwise, there must exist some $j \not = i$ such that the marginal probability of $j$ will be different between the circumstances that $i$ is protected and not protected. This concludes the proof. 
\end{proof}
Proposition \ref{thm:EntroLoss} shows that, conditioned on $X_k$, the entropy sum of all other $X_i$'s strictly decreases in expectation. Therefore,  the correlations among targets are intrinsic and inevitable. 

\vspace{1mm}

\noindent {\bf Dilemma of Traditional Algorithms}. Traditional security game algorithms, e.g., comb sampling \cite{Tsai2010} and strategy/column generation \cite{Jain2010,Bosansky2015}, are \emph{designed} to generate mixed strategies of  small support for the sake of computational efficiency.  Unfortunately, such small-support mixed strategies typically induce high correlation among the protection status of targets, and thus exacerbate the curse of correlation. This is illustrated by the above example as well as in our later experiments.

\subsection{Why Maximizing Entropy?}
To tackle the curse of correlation in security games, the ideal approach is to come up with an accurate model to capture  the attacker's partial observation (also referred to as an information \emph{leakage model} henceforth for convenience), and then solve the model to obtain the defender's optimal defending strategy. For example, one natural choice is to represent the problem as an extensive-form game (EFG) since information sets in EFGs naturally capture what the attacker knows at each time step.  However we argue that these approaches suffer from several critical drawbacks.  

\begin{enumerate}
	\item[] {\bf  \rom{1}. Unavailability of an Accurate Leakage Model.} The attacker's choice of target monitoring depends on many hidden factors, thus is highly unpredictable. Therefore, it is  typically very difficult to know which targets are leaking information -- otherwise the defender could have resolved the issue in the first place via other approaches. As a result, it is usually not possible to obtain an accurate leakage model. However, as we will illustrates in our experiments, optimizing over an inaccurate leakage model can even be harmful to the defender compared to doing nothing.   


\item[] {\bf \rom{2}. Scalability and Computational Barriers.} Even if the defender has an accurate leakage model, computing the optimal defender strategy against the leakage model is intractable generally.  As we mentioned in the related work, even in the simplest possible model -- zero-sum games, no spatio-temporal structures, no scheduling constraints and a single target leaking information-- computing the optimal defender strategy is NP-hard  \cite{Xu15}. Obviously, the problem only becomes more difficult  in our (more complicated) spatio-temporal settings with scheduling constraints. If we adopt the EFG representation, the size of the information sets increases exponentially in the number of time steps, number of patrollers and possible leaking targets. For example, a game with $9$ targets, $9$ time units and $2$ patrollers has  $O(9^{18}\cdot 2^9)$ information sets, which cannot be solved by any state-of-the-art EFG solvers. However, our algorithm can compute the max-entropy implementation for such an instance in a few seconds in our experiments.    


\item[] {\bf \rom{3}. Vulnerability to Attacker's Strategic Manipulations.} Another  concern about any optimal solution tailored to a specific leakage model is that such a solution may be easily ``gamed'' by the attacker. In particular, the optimal solution naturally biases towards the leaking targets by assigning more security forces to these targets. This however opens the door for the attacker to strategically manipulate the defender's belief on leaking targets, e.g., by intentionally spreading misleading information, with the goal of shifting the defender's defense away from the attacker's prime targets. As we show in our experiments, this could cause significant loss to the defender. 
\end{enumerate}

\vspace{1mm} 
\noindent {\bf Entropy maximization -- a more robust solution.} These barriers motivate our adoption of a more  robust (though inevitably more conservative) approach. Particularly, we propose to first compute the optimal defender strategy assuming no leakage and then adopt the max-entropy implementation of its marginal vector. 
%
Our choice of max entropy is due to at least three reasons. First, the max-entropy strategy is the most random, thus unpredictable, defender strategy. When the defender is uncertain about which target is leaking information (the setting we are in), we believe that taking the most random strategy is one natural choice. Second, the max-entropy distribution exhibits substantial approximate stochastic independence among the protection statuses of  targets\footnote{This is widely observed in practice, and also theoretically proved in some settings, e.g., matchings \cite{Kahn1997}.}, so that the protection status of any leaking targets does not carry much information about that of others.  Third, as we will illustrate in Section \ref{sec:exp}, empirically, our max entropy approach performs extremely well in comparisons with several other alternatives; in fact, in some settings,  it achieves a solution quality that is even close to the optimal defender utility under no leakage! Given such significant empirical results, entropy maximization  clearly stood out as a powerful approach to  address information leakage.

From a practical perspective, the max-entropy approach also enjoys several advantages. First,  it does not require a concrete leakage model. Instead, it seeks to reduce the overall correlation among the statuses of all targets, thus serves as a robust solution. Second, this approach is easily ``compatible" with any current deployed security systems since it does not require any change to previously deployed algorithms while only adds additional randomness (in some sense, this is a \emph{strictly} better solution than previous ones). This is particularly useful in domains where re-building a new security system is not feasible or too costly. 

Efficient computation remains a challenge for our approach due to the widely-known difficulty of computing the max-entropy distribution over combinatorial structures. Fortunately, we show  in the next two sections that they admit efficient algorithms in several important and well-motivated security settings, as well as good heuristics in general.  

\section{Efficient Max-Entropy Implementation  for Spatio-Temporal Security Games}
\label{sec:MaxEntro} 
In many applications, e.g., wildlife protection,  the security game is played out in space and time \cite{Gatti1,Fang2013,Yin2015}, which is termed the \emph{spatio-temporal} security game. In these domains,  the defender needs to move patrollers as time goes on. Due to the correlation among the patroller's consecutive moves, these games are more likely to suffer from CoC. In this section, we seek to mitigate CoC by entropy maximization. In particular, given any marginal vector (obtainable from any algorithm solving the original game), we seek to sample its max-entropy implementation. 

Like most previous work, we focus on discretized spatio-temporal security games. Such a game is described by a $T\times N$ grid graph 
\begin{figure}[H]
\centering
		\includegraphics[bb=200bp 120bp 680bp 500bp,clip,scale=0.28]{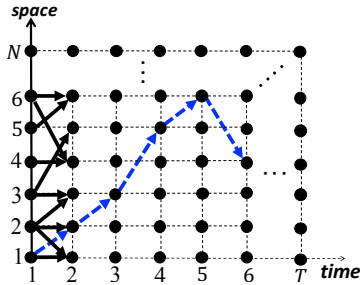}%
	\caption{\label{fig:SpatTemp}Structure of a Spatio-Temporal Security Game}
\end{figure}  
$G=(V,E)$ indicating a problem with $N$ cells and $T$ time layers (see Figure \ref{fig:SpatTemp}). We use $v_{t,i}$ to denote a grid node, representing cell $i$ at time $t$. Each $v_{t,i}$ is  treated as a target, so there are $n = T\times N$ targets.  The directed edges between each consecutive time layers denote the patroller's feasible moves between cells within consecutive time layers.
Such feasibility usually incorporates speed limit, terrain and other real-world constraints.  Figure \ref{fig:SpatTemp} depicts some  feasible moves between time layer $1$ and $2$. A feasible patrol path is a path in $G$ starting from time $1$ and ends at time $T$ (e.g., the dashed path in Figure \ref{fig:SpatTemp}). Note that there are exponentially many patrol paths. We assume that the defender has $k$ \emph{homogeneous} patrollers, so a defender pure strategy corresponds to the \emph{set of nodes} covered by $k$ feasible patrol paths.  
We start by proving that the max-entropy implementation problem for spatio-temporal security games is computationally hard.

\begin{theorem}\label{thm:hard}
	It is \#P-hard to sample the max-entropy implementation for spatio-temporal security games even when there are \emph{two} time layers. 
\end{theorem}
\begin{proof}
	
	When there are two time steps, the game structure  corresponds to a bipartite graph ($T=2$ in Figure \ref{fig:SpatTemp}). It is important to notice that a pure strategy here does \emph{not} simply correspond to a bipartite matching of size $k$, therefore we cannot reduce from the problem of counting size-$k$ matchings. This is because the selected $k$ edges are allowed to share nodes.  Moreover,  our definition of a pure strategy is the set of covered targets, while not the edges themselves. In fact, sometimes one pure strategy can be achieved by different set of $k$ edges. 
	
	To prove the theorem, we reduce from the problem of counting bases of a transversal matroid, which is known to be \#P-complete \cite{Colbourn1995}. Given any bipartite graph $G=(L\cup R, E)$ with $|L|=k$, $|R| = n$ and $k \leq n$, any set $T \subseteq R$ is an independent set of the transversal matroid $\M (G)$ of $G$ if there exists a matching of size $|T|$ in the subgraph induced by $L \cup T$; Such a $T$ is a base if $|T| = |L| = k$. 
	
	Given any bipartite graph $G=(L\cup R, E)$ with $k = |L| \leq |R| = n$, we reduce counting bases of the transversal matroid $\M(G)$ to the max-entropy implementation for the two-time-layer spatio-temporal security games on graph $G$.\footnote{Though the definition of spatio-temporal security games requires that each time layer has the same number of nodes, this requirement is not essential since one can always add  \emph{isolated} nodes to each time layer to equalize the number of nodes.} Let $\S_{2k}$ denote the set of pure strategies that cover exactly $2k$ nodes. We first reduce counting bases of $\M(G)$ to counting $\S_{2k}$. This simply because if a pure strategy covers $2k$ nodes, it must cover all $k$ nodes in $L$ and another $k$ nodes in $R$, and these $2k$ nodes are matchable. It is not hard to see that elements in $\S_{2k}$  and $\M (G)$ are in one-to-one correspondence.

	Since counting reduces to generalized counting, we finally reduce generalized counting over set $\S_{2k}$ to max-entropy implementation of the following special subset of marginal vectors $\X_{2k} = \{\mathbf{x} \in [0,1]^{k+n}: x_{1,i} = 1, \forall i \in L; \sum_{i=1}^n x_{2,i} = k  \}$. It is easy to see that any mixed strategy that matches a marginal vector $\mathbf{x} \in \X_{2k}$ must support on $\S_{2k}$. Therefore,  when considering max-entropy implementation of any  $\mathbf{x} \in \X_{2k}$, we can w.l.o.g. restrict the set of  pure strategies to be $\S_{2k}$. By Lemma \ref{lem:EntroSample}, generalized counting over $\S_{2k}$ reduce to max-entropy implementation for  any  $\mathbf{x} \in \X_{2k}$.  

\end{proof} 
Theorem \ref{thm:hard} suggests that it is unlikely that there is an efficient algorithm for the max-entropy implementation problem in general. Moreover, there is no known \emph{polynomial size} compact formulation for computing the max-entropy distribution,   
thus we cannot utilize the state-of-the-art optimization software to tackle the problem. Nevertheless, we show that the max-entropy approach, termed \en \,  henceforth, can be efficiently implemented in two well-motivated special settings.

\vspace{-1mm}
\subsection{\en \, for  Setting 1: Small Number of  Patrollers}\label{sec:const}
\vspace{-1mm}
In many applications, the defender only possesses a small number of patrollers. For example, in wildlife protection, the defender usually has only one or two patrol teams
at each patrol post \cite{Fang2016}; the US Coast Guard uses two patrollers to protect Staten Island ferries \cite{Fang2013}. In this section, we show that when the number  of patrollers is small (i.e., can be treated as a constant), the max entropy implementation can be sampled efficiently. 


\begin{theorem}\label{thm:WildCount}
	When the number of patrollers is a constant, there is $\poly(N,T)$ time algorithm for the max-entropy implementation in spatio-temporal security games. 
\end{theorem}

\begin{proof}
	Theorem \ref{thm:WildCount} is proved by developing an efficient algorithm for the  generalized counting version of the problem (which is all we need to sample a max-entropy implementation by Lemma \ref{lem:EntroSample}). For ease of presentation, our description focuses on the case with two patrollers, though it  easily generalizes to a constant number of patrollers.  We propose a dynamic program (DP) that exploits the natural chronological order of the targets along the temporal dimension.\footnote{Dynamic programming is widely used in counting problems. See, e.g., \cite{DP1,DP2} as well as the remarks in \cite{DP3}. The novel parts  usually lie at careful analysis of the problem to uncover the proper structure for DP.} 
	
	Let us call a pure strategy $2$\emph{-path}, since each of the two patrollers takes a path on $G$.  Our goal is to compute the weighted count of all $2$-paths in grid graph $G$, each weighted by  the product of the node weights it traverses. 
	Our goal is to compute the weighted count of combinations of two paths in $G$ (one for each patroller), where the weight is the product of the node weights that the two paths traverse. 
	Let $\{\alpha_{t,i} \}_{t \in [T], i\in [n]}$ be any given weight set.  Obviously, the counting problem is easy if $T=1$, i.e., only one time layer. Our key observation is that the solution for $T=t$ can be constructed by utilizing the solutions for $T=t-1$.  Particularly, for any $1 \leq i \leq j \leq N$, we use $\ddp (i,j;t)$ to denote the solution to the counting problem  restricted to the truncated graph with only time layers $1,2,..,t$, and satisfies that the two patroller must end at cell $i,j$ at time $t$. Observe that $\ddp (i,j;1) = \alpha_{1,i} \alpha_{1,j}$ when $i \not = j$ and $\ddp(i,i;1) = \alpha_{1,i}$. We remark that our algorithm description in the main section intentionally ``omitted" the trivial conner case of $i = j$ for ease of presentation.   We then use the following update rule for $t \geq 2$: $\ddp (i,j;t)=$
	\[ \begin{cases} 
	\alpha_{t,i} \alpha_{t,j} \cdot \sum_{(i',j') \in \mbox{pre}(i,j)} \ddp(i',j';t-1) &  \text{ if }i < j \\
	\alpha_{t,i}\cdot \sum_{(i',j') \in \mbox{pre}(i,j)} \ddp (i',j';t-1) & \text{ if } i = j
	\end{cases}
	\]
	where	 $\mbox{pre} = \{(i',j'): i' \leq j' \, s.t. \, v_{t-1,i'},v_{t-1,j'}$ can reach $v_{t,i},v_{t,j}$  $\}$ is essentially the set of all pairs of nodes that can reach $v_{t,i},v_{t,j}$.  Note that the solution to the generalized counting problem is $\sum_{i \leq j} \ddp(i,j;T)$. 
	The correctness of the algorithm follows from the observation that if the two patrollers are at $v_{t,i}$ and $v_{t,j}$, they must come from $v_{t-1,i'}$ and $v_{t-1,j'}$ for certain $(i',j') \in \mbox{pre}(i,j)$. The updating rule simply aggregates all such choices. The algorithm runs in $\poly(N,T)$ time.

\end{proof}

\vspace{-1mm}
\subsection{\en \, for Setting 2: Air Marshal Scheduling}
\label{sec:AirMarshal}
\vspace{-1mm}
One important task faced by the Federal Air Marshal Service (FAMS) is to schedule air marshals to protect \emph{international} flights. In this case, the schedule of each air marshal is a \emph{round trip} \cite{Kiekintveld2009}, and the game can be described as a spatio-temporal security game with two time layers. In this section, we show that the max-entropy implementation  can be sampled efficiently in this setting, \emph{despite} that  the \#P-hardness in Theorem \ref{thm:hard} holds for the case of two time layers. Our algorithm crucially exploits certain ``order'' structure of the air marshal's schedules. 

We start by describing the problem.  FAMS seeks to allocate $k$ homogeneous air marshals to protect round-trip international flights originating from certain domestic  city  to different outside cities. These round-trip
\begin{wrapfigure}{r}{69mm}
	\vspace{-4mm}
	\begin{center}
		\includegraphics[bb=120bp 160bp 660bp 450bp,clip,scale=0.25]{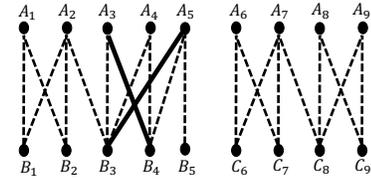}%
	\end{center}
	\vspace{-4mm}
	\caption{Consistent round-trip flights between a domestic city and two outside cities. }\label{fig:FAMS}
	\vspace{-3mm}
\end{wrapfigure} 
 flights constitute a bipartite graph $G=(A\cup B, E)$ in which nodes in $A$ [$B$] correspond to all originating [return] flights; $e=(A_i, B_j) \in E$ iff $e$ forms a \emph{consistent} round trip.
Note that $G$ is  a union of multiple isolated smaller bipartite graphs, each containing all flights between two cities (see Figure \ref{fig:FAMS}).  This is because any flight from $a$ to $b$ can never form a round trip with a flight from $c$ to $a$. We will call each isolated bipartite graph a \emph{component}.   
The following are natural constraints on the structure of $G$ (which makes the problem more tractable than general settings):  $(A_i,B_j) \in E$ if (1) the destination of $A_i$ is the departure city of $B_j$; (2) the arrival time of $A_i$, denoted as $arr(A_i)$, and the departure time of $B_j$, denoted as $dep(Bj)$, satisfy $dep(B_j ) - arr(A_i) \in [T_1,T_2]$ for constants $T_2 >T1 > 0$.

\begin{theorem}\label{thm:AirMaxEntropy}
	The max-entropy implementation for the FAMS problem with round trips can be sampled in polynomial time. 
\end{theorem}

\subsection*{Proof of Theorem \ref{thm:AirMaxEntropy}} 
By Lemma \ref{lem:EntroSample}, we only need to  design an efficient algorithm for the generalized counting version of the problem. Let $G = (A \cup B, E)$ denote the bipartite graph for the air marshal scheduling problem, $|A| = n_1, |B| = n_2$.  Recall that $G$ is  a union of multiple isolated components, each containing all flights between $a$ and another city (see Figure \ref{fig:FAMS}).  
Within each component, we sort the flights in $A$  by their \emph{arrival} time and flights in $B$ by by their \emph{departure} time.  

We now show that the generalized counting over the set of defender pure strategies admits a polynomial time algorithm. 
Our algorithm crucially exploits the following ``order''   structure.

\begin{definition}\label{def:order} [Ordered Matching]
	In a bipartite graph $G=(A \cup B, E)$, a matching  $M = \{e_1,...,e_k \}$  is called an \emph{ordered matching} if for any edge  $e = (A_i,B_j)$ and $e' = (A_{i'},B_{j'})$ in $M$, either $i > i'$, $j > j'$ or $i < i'$, $j < j'$.  
\end{definition}
Visually, any two edges $e,e'$ in an ordered matching satisfy that $e$ is either ``above'' or ``below'' $e'$ -- they do not cross. 

Since each flight has at most one air marshal, any assignment of air marshals must correspond to a matching in $G$.   However, a pure strategy $S$ -- i.e., a set of covered flights -- can be accomplished by different matchings. For example, the set $S = \{ A_1, A_2, B_1, B_2\}$ in Figure \ref{fig:FAMS} can be achieved by the matching $\{(A_1, B_1), (A_2, B_2)\}$ or the matching $\{(A_1, B_2), (A_2,B_1) \}$. Nevertheless, only the matching $\{(A_1, B_1), (A_2, B_2)\}$  is ordered. It turns out that pure strategies and \emph{ordered} $k$-matchings are in one-to-one correspondence.
\begin{lemma}\label{lem:order}
	In the air marshal scheduling problem, there is a way to order flights in $A$ and $B$  so that pure strategies and size-$k$ ordered matchings  are in one-to-one correspondence. 
\end{lemma}
\begin{proof}
	It is easy to see that any ordered $k$-matchings  corresponds to one pure strategy. We prove the reverse. Given any pure strategy $S$ consisting of $2k$ flights, let $\tilde{E} = \{e_1,...,e_k\}$ be any matching that results in $S$. We claim that if there exist two edges $e,e'  \in \tilde{E}$ with $e = (A_i,B_j)$ and $e' = (A_{i'},B_{j'})$ such that $i > i'$ and $j < j'$, then $(A_i,B_{j'})$ and $(A_{i'},B_j)$ must also be edges in $E$. 
	Since $e,e' \in \tilde{E}$, we must have $T_1 < dep(B_j) - arr(A_i) < T_2 $ and $T_1 < dep(B_{j'}) - arr(A_{i'}) < T_2 .$  Since flights in $A$ are ordered increasingly  by arrival time and flights in $B$ are ordered increasingly by departure time, we have $arr(A_i) \geq  arr(A_{i'})$ and  $dep(B_j) \leq dep(B_{j'})$. These inequalities imply $dep(B_{j'}) - arr(A_i) \leq dep(B_{j'}) - arr(A_{i'}) \leq T_2 $ and $dep(B_{j'}) - arr(A_i) \geq dep(B_{j}) - arr(A_{i}) \geq T_1 $, therefore $(A_i,B_{j'}) \in E$. Similarly, one can show that $(A_j,B_{i'}) \in E$. 
	
	As a result, we can adjust the matching by using the edges $(A_i,B_{j'})$ and $(A_{i'},B_j)$ instead. Such adjustment can continue until the matching becomes ordered. The procedure  will terminate within a finite time by a simple potential function argument, with potential function $f(\tilde{E}) = \sum_{e = (A_i,B_j) \in \tilde{E}} |i-j|^2$. In particular, the above adjustment always strictly decreases the potential function since $|i-j|^2 + |i' - j'| ^2 > |i - j'|^2 + |i' - j|^2$ if $i > i'$ and $j < j'$. The adjustment will terminate with an ordered matching. It is easy to see that the ordered matching is unique, concluding our proof.  
\end{proof}

Lemma \ref{lem:order} provides a way to reduce  generalized counting over the set of pure strategies  to generalized counting of size-$k$ ordered matchings. In particular, given any set of  non-negative weights $\mathbf{\alpha} \in \RR^{n_1+n_2}_+$, we define edge weight $w_e = \alpha_{A_i} \alpha_{B_j}$ for any $e = (A_i, B_j) \in E$. Therefore, the weight of any pure strategy equals the weight of the corresponding size-$k$ ordered matching with edge weights $w_e$'s.  It turns out that generalized counting of size-$k$ ordered matchings admits an efficient algorithm.

The main idea is to dynamically compute the generalized sum of size-$k$ ordered matchings according to some ``order" of the bipartite graphs. More specifically, 
define $E_{l,r} \subseteq E$ to be the set of edges that are ``under" $A_l$ and $B_r$, where $A_l \in A,B_r\in B$. Formally, any $e = (A_i,B_j) \in E$ is in $E_{l,r}$ iff $i \leq l, j \leq r$. We build the following dynamic programing table $\ddp (l,r;d)=\sum_{M: \, M \subseteq E_{l,r},|M|=d} w_{M}$, in which  $\ddp (l,r;d)$ is the sum of  the weight of all size-$d$ ordered matchings with edges in $E_{l,r}$. Now, to compute $\ddp (l,r;d)$,  we only need to enumerate all the possibilities of the uppermost edge in the ordered matching, given that $\ddp (i,j;d)$s are known for $i < l$ and $j < r$.  This can be done by a dynamic program (Algorithm \ref{alg:CountMatch}).   

\begin{algorithm}[ht]
	\begin{algorithmic}[1]
		\REQUIRE : $G=(A\cup B, E)$;  $w_e \geq 0$ for any $e \in E$.
		
		\ENSURE : $\sum_{M: \, |M|=d} w_{M}$ \hspace{4mm}\% $M$ is an ordered matching
		\vspace{1mm}
		
		\STATE {\bf Initialization}: $\ddp(l,r;0) = 1$ for $l = 0,..,n_1, r= 0,..,n_2$; $\ddp(0,r;d) = \ddp(l,0;d) = 0$ for all $d \geq 1, l =0,1,,...,n_1, r= 0,1,...,n_2$.  
		
		\STATE Update:  for $ d = 1,...,k, l=2,...,n_1, \, r = 2,...,n_2$:
		\begin{eqnarray*}
			& & 	\ddp(l,r; d) = T(l-1,r-1;d) +  \\
			& & \quad  \sum_{ \substack{e = (A_i,B_j) \in E_{l,r}\\ \text{ s.t. } i = l \text{ or }  j = r }}  w_e \cdot \ddp(i-1,j-1;d-1).
		\end{eqnarray*}
		
		\RETURN $\ddp(n_1,n_2;k)$.
	\end{algorithmic}
	\caption{Generalized Counting of Ordered $k$-Matchings	\label{alg:CountMatch}}
\end{algorithm}

\section{\ca: A Heuristic Sampling Algorithm} \label{sec:card}
In this section, we propose a heuristic sampling algorithm that implements any given implementable marginal vector $\mathbf{x}$ and potentially achieves high entropy. This algorithm works for \emph{general} security games (not only spatio-temporal ones), and is \emph{computationally efficient} as long as the underlying security game can be solved efficiently. 

At a high level, our idea is to design a \emph{randomized} implementation for  the celebrated Carath\'{e}odory's theorem, which is about the following \emph{existence} statement: for any bounded polytope $\P \subseteq \RR^n$ and any $\mathbf{x} \in \P$, there exists at most $n+1$ vertices of $\P$ such that $\mathbf{x}$ can 
be written as a convex combination of these vertices. Interpreting $\P$  as the convex   hull of defender pure strategies, this means any defender mixed strategy, i.e., a point in $\P$, can be decomposed as a distribution over at most $n+1$ pure strategies ($n$ is the number of targets).   We turn this existence statement into an efficient randomized algorithm, named {\bf CA}rath\'{e}odory {\bf R}andomized {\bf D}ecomposition (\ca). 

Consider any  polytope $\P = \{\mathbf{z}: A \mathbf{z} \leq \mathbf{b}; M \mathbf{z} = \mathbf{c} \}$ explicitly represented by polynomially \begin{wrapfigure}{r}{63mm}
	\begin{flushright}
		\vspace{-2mm}
		\includegraphics[bb=30bp 160bp 450bp 470bp,clip,scale=0.33]{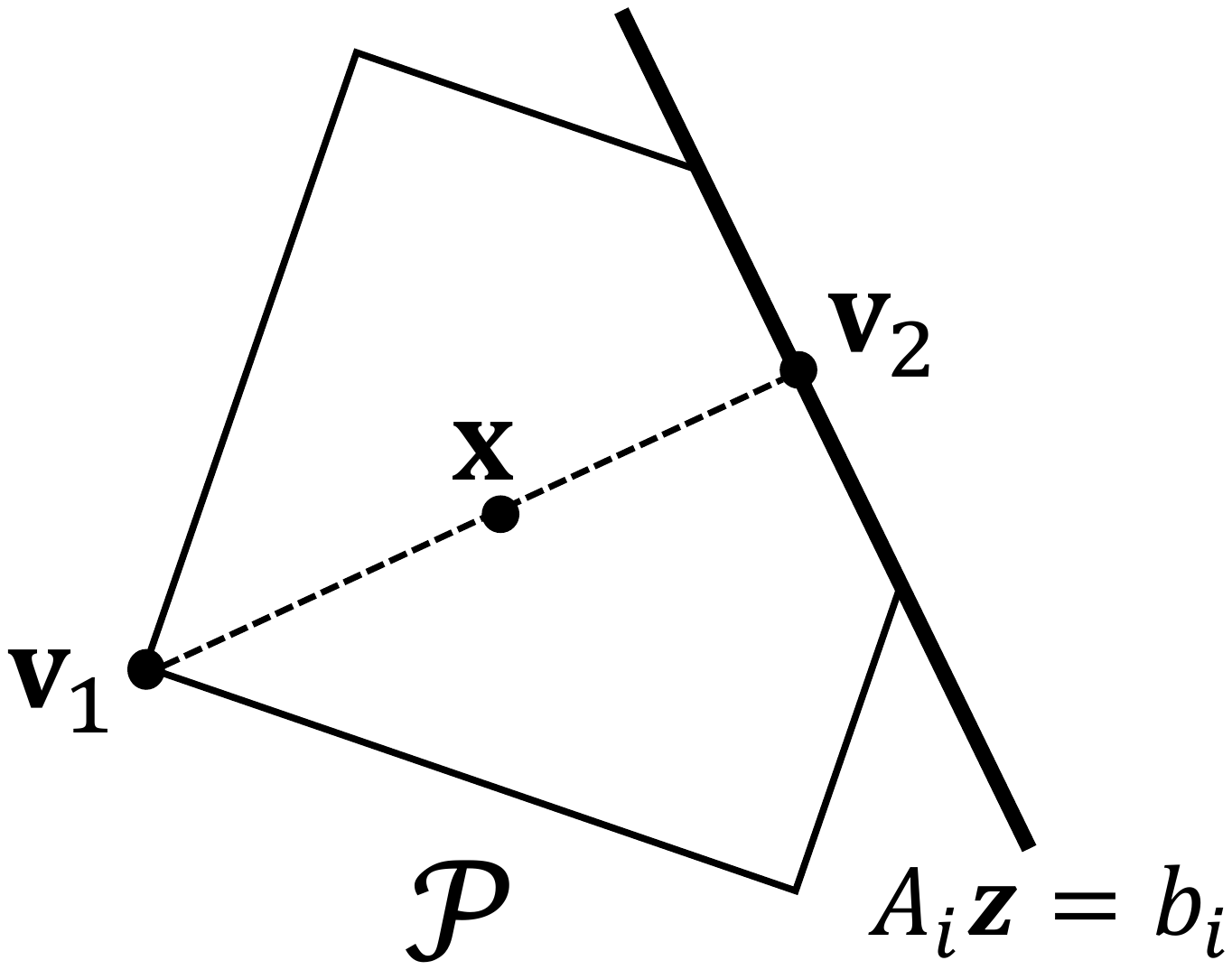}%
		\vspace{-1mm}
		\caption{\label{fig:card} \small \ca \, Decomposition.}
	\end{flushright}
	\vspace{-2mm}
\end{wrapfigure} 
many linear constraints, and any $\mathbf{x} \in \P$. We use $A_i,b_i$ to denote the $i$'th row of $A$ and $b$ respectively;  $A_i \mathbf{z} = b_i$ is a \emph{facet} of $\P$.  Geometrically, \ca \, randomly picks a vertex $\mathbf{v}_1 = \arg \max_{\mathbf{z} \in \P}  \langle \mathbf{a}, \mathbf{z} \rangle$ for a  linear objective $\mathbf{a} \in [0,1]^n$ chosen \emph{uniformly at random}. \ca \, then ``walks along'' the ray that originates from $\mathbf{v}_1$ and points to $\mathbf{x}$,  until it  crosses a facet of $\P$, denoted by $A_i \mathbf{z} = b_i$, at some point $\mathbf{v}_2$ (see illustrations in Figure \ref{fig:card}). Thus, $\mathbf{x}$ can be decomposed as a convex combination of $\mathbf{v}_1, \mathbf{v}_2$. \ca \,  then treats $\mathbf{v}_2$ as a new $\mathbf{x}$ and   decompose it within the facet $A_i \mathbf{z} = b_i$ recursively until $\mathbf{v}_2$ becomes a vertex. Details are presented in Algorithm \ref{alg:Caratheodory}. 

\begin{algorithm}
	\begin{algorithmic}[1]
		\REQUIRE $\P = \{\mathbf{z} \subseteq \RR^n: A \mathbf{z} \leq \mathbf{b}; M \mathbf{z} = \mathbf{c} \}$  and $\mathbf{x} \in \P$
		\ENSURE  $\mathbf{v}_1,...,\mathbf{v}_k$ and $p_1,...,p_k$ such that $\sum_{i=1}^k p_i  \cdot \mathbf{v}_i = \mathbf{x}$. 
		\vspace{1mm}		
		\IF{rank($M$) = $n$}
		\STATE Return the unique point $\mathbf{v}_1$ in $\P$ and $p_1 = 1$. 
		\ELSE
		\STATE	Choose $\mathbf{a} \in [-1,1]^n$ \emph{uniformly at random}
		\STATE  Compute $\mathbf{v}_1 = \arg \max_{\mathbf{z} \in \P}  \langle \mathbf{a}, \mathbf{z} \rangle $.
		\STATE	Compute $t = \min_{i: A_i (\mathbf{x} - \mathbf{v}_1) > 0} \frac{b_i - A_i \mathbf{x}_i}{A_i (\mathbf{x} - \mathbf{v}_1)}$; Let $i^*$ be the row obtaining $t$, and $\P' = \{\mathbf{z} \in \P: A_{i^*} \mathbf{z} = b_{i^*} \}$.  
		\STATE	$\mathbf{v}_2 = \mathbf{x}+t(\mathbf{x} - \mathbf{v}_1)$; \,  $p_1 = \frac{t}{t+1}, \, \, \,   p_2 = \frac{1}{1+t}$. 
		\STATE $[V',\mathbf{p}'] = \text{\ca}(\mathbf{v}_2,\P')$
		\RETURN $V = (\mathbf{v}_1, V' )$ and $\mathbf{p}=(p_1, p_2 \times \mathbf{p}' )$. 
		\ENDIF
	\end{algorithmic}
	\caption{\ca }
	\label{alg:Caratheodory}
	\vspace{-0mm}
\end{algorithm}

A crucial ingredient of \ca \,  is that each vertex $\mathbf{v}_1$ is the optimal vertex solution to a \emph{uniformly random} linear objective. Recall that the max-entropy distribution over any given support under no constraints is the uniform distribution. The  intuition underlying \ca \,  is that these randomly selected vertices will inherit the high entropy of their linear objectives. Notice that the decomposition generated by \ca \,  is different in each execution due to its randomness, therefore the strategies generated by \ca \,  are sampled from a very large support.

\section{Experimental Evaluations}\label{sec:exp}

\subsection{Comparisons with Optimal Approach}\label{subsec:CompOpt}
To our knowledge,  \cite{Xu15} is the only work that directly considers the computation of the optimal defender strategy against a specific leakage model and proposes an \emph{exponential} time algorithm for a \emph{zero-sum} security setting with no resource allocation constraints and one leaking target. They considered two information leakage models: 1. the defender has a probabilistic belief about the leaking target (\emph{probabilistic leakage}); 2. the leaking target is adversarially chosen by the attacker to minimize the defender's utility (\emph{adversarial leakage}).  \emph{These (simpler) settings are not what we focus on}, but our algorithms can be applied here. Therefore,  we experimentally test max-entropy algorithms in this setting, with the goal of examining how our approach performs compared to the rigorously optimal solution of \cite{Xu15}, termed \ri,  


We compare three algorithms: \ri, \en, and \ca \, (\en, and \ca \, are from Section \ref{sec:MaxEntro} and \ref{sec:card} respectively). All algorithms are implemented by MATLAB 2016 and run on a single machine with 2.7 GHz Intel Core i5 GPU and 8 GB memory. Figure \ref{exp:time} compares the running time of these algorithms. \en \, and \ca \,  are much more scalable than \ri.  Particularly, \ri \,  scales only to 25$\sim$30 targets due to its exponential running time, while \en \,  and \ca \, can solve instances of $300$ targets within $2$ minutes in our experiments. 

\begin{figure*}
	\centering
	\subfigure[Running Time Comparison]{
		\label{exp:time}
		\includegraphics[bb=0bp 0bp 550bp 320bp,clip,scale=0.25]{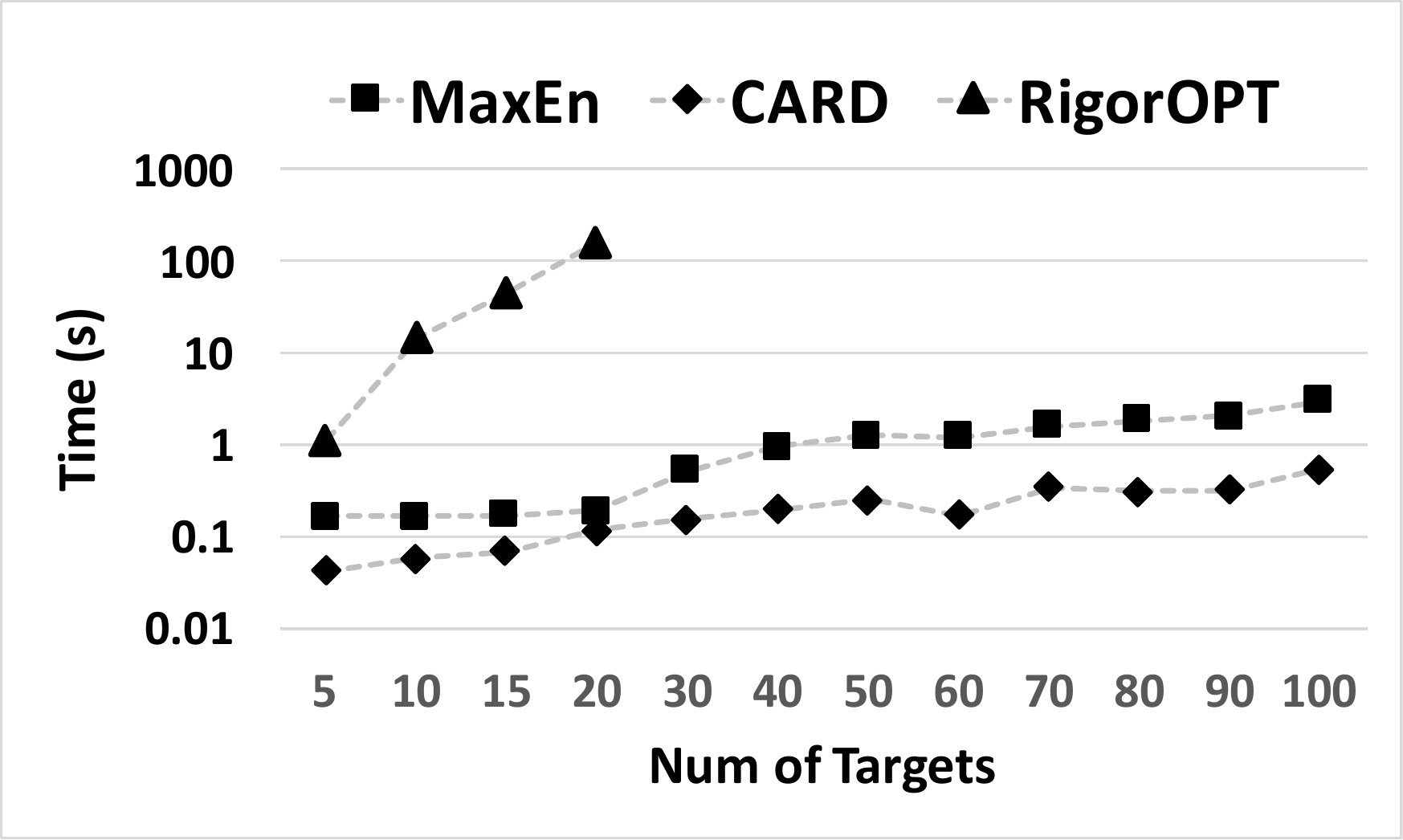}	
	}
	\subfigure[Probabilitic Leakage]{
		\label{exp:probL}
		\includegraphics[bb=0bp 0bp 550bp 320bp,clip,scale=0.25]{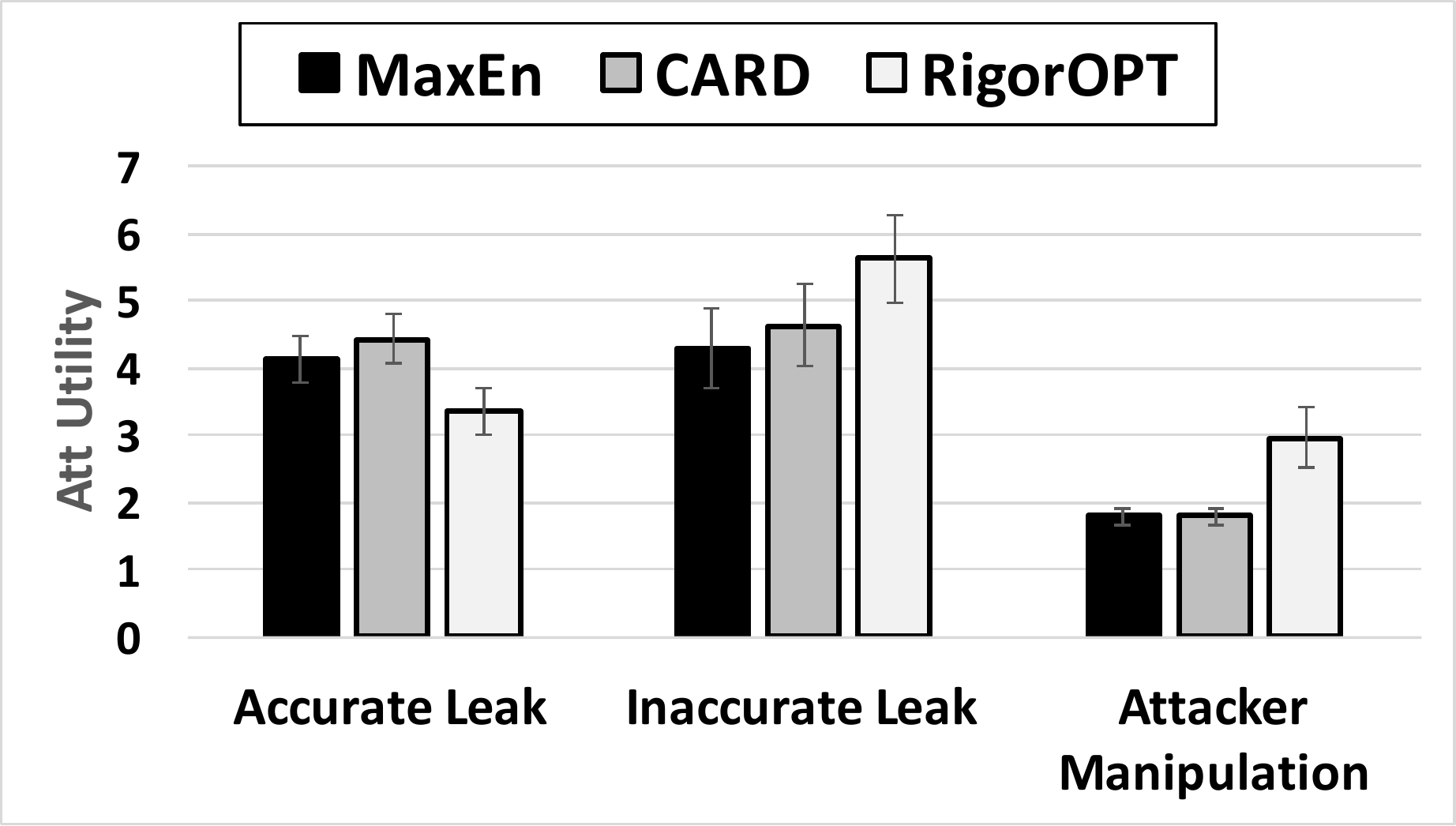}
	}
	\subfigure[Adversarial Leakage]{
		\label{exp:advL}
		\includegraphics[bb=0bp 0bp 550bp 320bp,clip,scale=0.25]{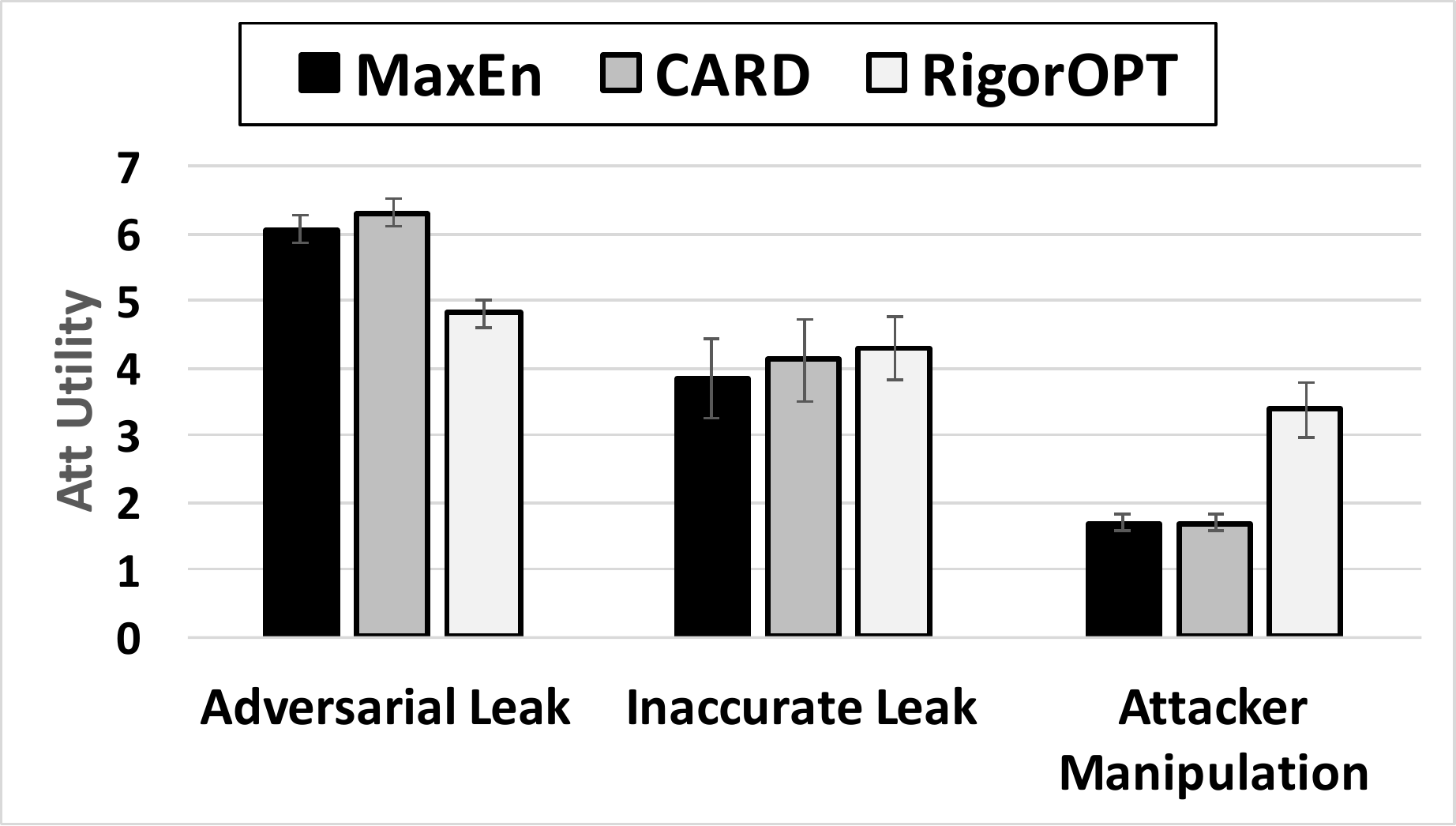}	
	}
	\vspace{-3mm}
	\caption{\small Comparisons of max-entropy algorithms and the rigorously optimal algorithm.}
	\label{fig:CompRigor}
	\vspace{-3mm}
\end{figure*}

Figure \ref{exp:probL} and \ref{exp:advL} compare the solution quality of these algorithms, with a focus on situations with model misspecification and attacker manipulations. Each column averages over 20 randomly generated \emph{zero-sum} games with  $20$ targets, $6$ patrollers and  utility drawn from $[-10,10]$ uniformly at random. 
Figure \ref{exp:probL} examines 3 scenarios. The ``\emph{Accurate  Leak}'' scenario assumes that the defender knows precisely which target leaks information; 
The ``\emph{Inaccurate Leak}'' scenario assumes that one target leaks information, but the defender does \emph{not} know exactly which one and only has an (inaccurate) probabilistic belief, over which she optimizes her strategy;  
The ``\emph{Attacker Manipulation}'' scenario assumes that there is \emph{no} leaking target, but the attacker misleads the defender to believe that one target leaks information. 
The $y$-axis is the attacker utility -- the lower, the better. As we can see, \ri \, only outperforms  \en \, and \ca \, in the ``\emph{Accurate  Leak}'' scenario (as expected since it is optimal), but is  outperformed otherwise. Figure \ref{exp:advL} is about similar scenarios but in the adversarial leakage model. Figure \ref{exp:probL} and \ref{exp:advL} show the robustness  of the max-entropy algorithms -- they are close to optimality when the model is indeed accurate, but are much more robust to  misspecification of leakage models and attacker manipulations than \ri.

\begin{figure*}[ht]
	\centering
	\subfigure[STST: Increase $T$]{
		\label{exp:TU}
		\includegraphics[width = 0.26\textwidth]{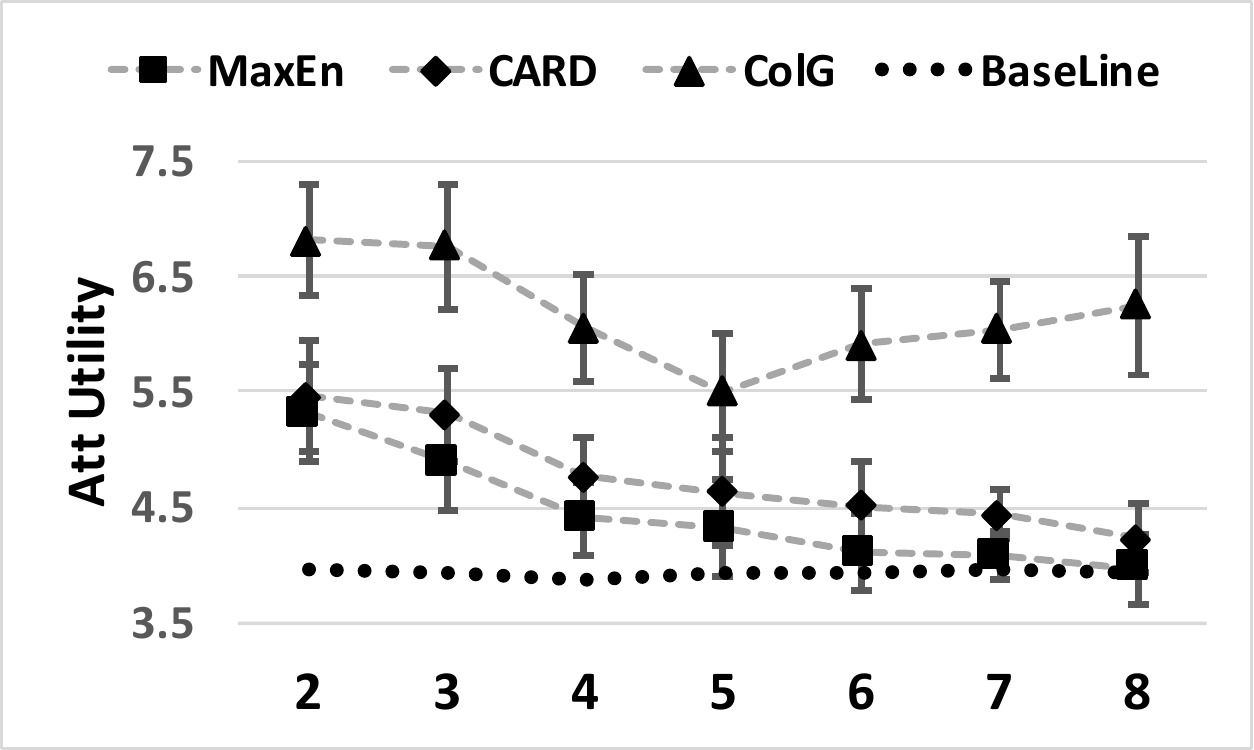}
	}
	\subfigure[STST: Increase \#MoT]{
		\label{exp:LU}
		\includegraphics[width = 0.292\textwidth]{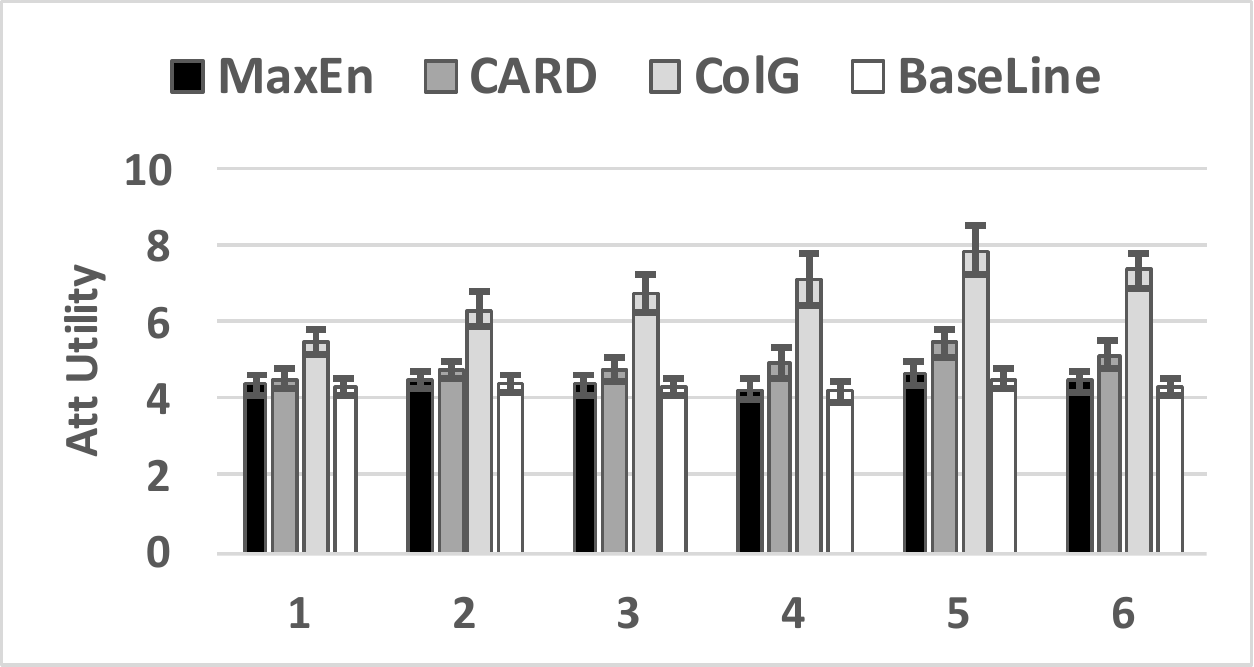}	
	}
	\subfigure[FAMS: Increase DtS]{
		\label{exp:A_DtS_defU}
		\includegraphics[width = 0.26\textwidth]{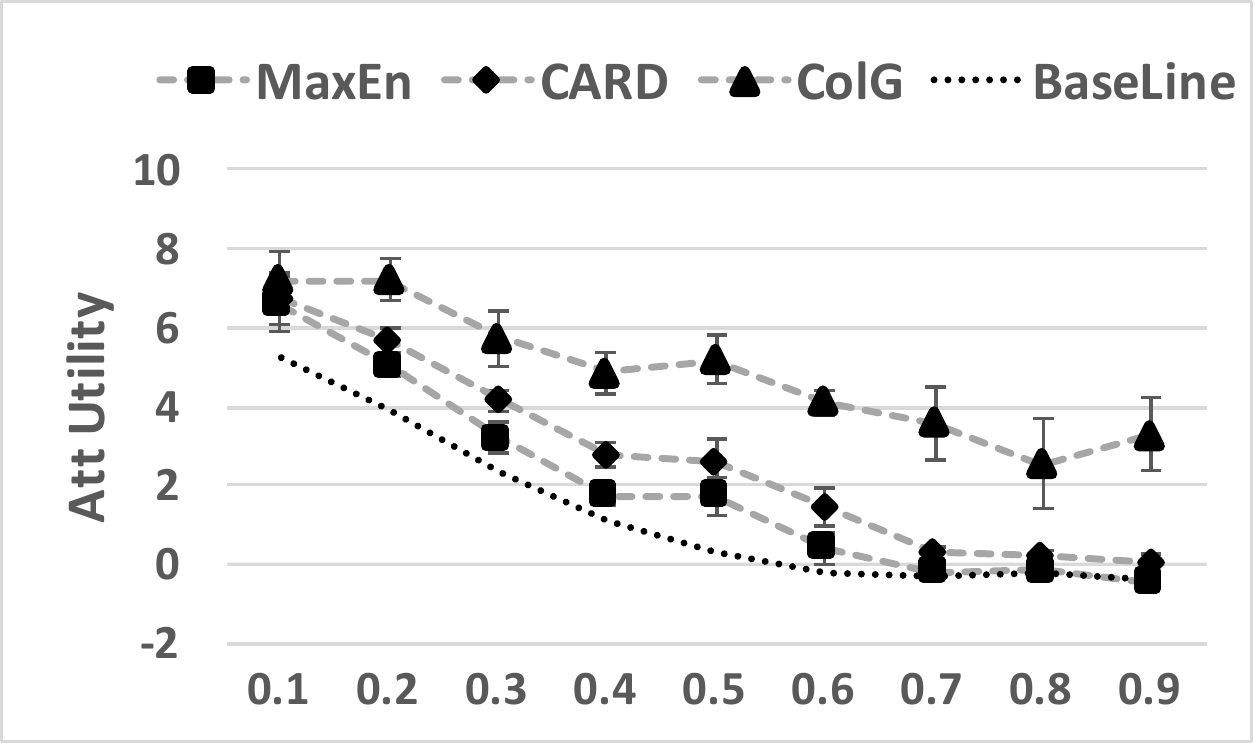}	
	}
	\vspace{-3mm}
	\caption{Comparisons of max-entropy and traditional algorithms in STST and FAMS settings.}
	\label{fig:wildlife}
	\vspace{-3mm}
\end{figure*}
\subsection{Comparisons with Traditional Algorithms in Spatio-Temporal Security Settings}\label{subsec:CompClassic}
We now experimentally compare \en,~\ca\, with traditional security game algorithms in spatio-temporal security settings. We are not aware of any previous algorithm that directly computes the optimal defender strategy against a particular leakage model, therefore the rigorously optimal solution is not available. We instead use a  ``harder" \texttt{BaseLine} which is the attacker utility assuming no leakage. This is the best (i.e., smallest) possible attacker utility.   The most widely used approach for solving large-scale security games is  the  \emph{column generation} technique (a.k.a., strategy/constraint generation \cite{Jain2010,Bosansky2015}). Particularly, we compare \en,  \ca \, with \ge~(the optimal mixed strategy computed via column generation\footnote{The  column generation technique is widely used in many security game algorithms. Though some security games use a compact linear program to directly compute the optimal marginal vector, the ultimate generation of a deployable mixed strategy still requires strategy generation techniques. In FAMS domain, \ge \,  is precisely  the ASPEN algorithm  \cite{Jain2010} -- the leading algorithm today for scheduling air marshals at scale.  }  assuming no leakage).   Note that all these algorithms use the same marginal vector. The goal is to test their robustness in presence of information leakage. 


All algorithms are tested in two settings: Spatio-Temporal Settings with Two patrollers (STST) and Federal Air Marshal Scheduling for round-trip flights (FAMS). 
All algorithms run efficiently (within $2$ minutes),  so we only compare their solution qualities in term of  attacker utilities -- the lower, the better.
Unless specifically mentioned, we always assume that the attacker can monitor two randomly chosen targets at time $t$=1 and attacks one target at the last time layer $ T$ ($T$=2 in FAMS).
All results are averaged over 20 zero-sum security games with utility drawn randomly from $[-10,10]$. We test the algorithms on zero-sum games because they are strictly competitive, therefore any information leaking to the attacker will benefit the attacker and hurt the defender. However, the effects of the curse of correlation (CoC) could be a mix of both good and bad in general-sum security games because there,  ``leaking" information to the attacker  could sometimes be  beneficial to the defender. This has been studied in previous work on strategic information revelation in security games \cite{Babinovich2015,Guo2017}. In zero-sum security games, however, it  was formally proved that any information to the attacker will hurt the defender. In this sense, zero-sum games serve as the best fit for studying harms of CoC. Previous research \cite{Alon2013} studying information leakage in normal-form games has also focused on zero-sum games. We remark that many security games, including several real-world applications in spatio-temporal domains\cite{Yin2012,Fang2013}, are indeed zero-sum.


In all our experiments,  \en \, significantly outperforms \ge; \ca \, is usually slightly worse than \en. We therefore mainly focus on analyzing other interesting findings here. Figure \ref{exp:TU} compares the algorithms in the STST setting by varying the number of time layers $T$, but fixing $N=9$.  Interestingly, when $T$ increases, \en \, and \ca \, approaches the \texttt{BaseLine}, i.e., the lowest possible attacker utility. This shows that in patrolling strategies of large entropy, the correlation between a patroller's initial and later moves gradually disappear as time goes on. 
This illustrates the validity of max-entropy approach for  mitigating CoC. 
In Figure \ref{exp:LU}, we fix $T = 9,N=9$, and  compare the algorithms by varying the number of Monitored Targets (\#MoT). Surprisingly, even when the attacker can monitor $6$ out of $9$ targets at $t=1$, \en \, and \ca \, are still close to \texttt{BaseLine}, while the performance of \ge \, gradually decrease as \#MoT increases.

Figure \ref{exp:A_DtS_defU} moves to the FAMS setting, and compares all algorithms by varying the deployment-to-saturation (DtS) ratio \cite{Jain2012}, which is $2k/n$ in FAMS ($n=60$ in Figure \ref{exp:A_DtS_defU}). DtS captures the \emph{fraction} of targets that can be covered in a pure strategy. It turns out that the higher the DtS ratio is, the worse \ge \, performs. 
This is because, with higher DtS,  \ge \, quickly converges to an optimal mixed strategy with very small support, since each pure strategy covers many targets. Unfortunately, such a small-support  strategy suffers severely from the curse of correlation. We noted that in FAMS games with 100 targets, \en,  \ca,  \ge\, use 99997, 2199, 53 pure strategies on average (\en\, samples 100,000 pure strategies in our experiments, and almost all of them are different). 



\section{Conclusions}
In this paper we formally study the curse of correlation (CoC) phenomenon in security games and show that it may cause significant loss to the defender if not addressed properly. We believe this study raises some important  caveats regarding the practical use of traditional security game algorithms. To mitigate CoC, we propose to adopt the defender strategy with maximum entropy,  elaborate the advantages of this approach and systematically justify it by empirical evaluations. Finally, we develop scalable algorithms for computing the max-entropy defender strategy in well-motivated applications.  One important open question is to develop other methods to address CoC in security games. 


\bibliographystyle{plain}
\bibliography{refer}

\end{document}